\documentclass[11pt,letterpaper]{article}
\usepackage[utf8]{inputenc}
\usepackage{amsmath}
\usepackage{amssymb}
\usepackage{amsthm}
\usepackage{amsfonts}
\usepackage{dsfont}
\usepackage{geometry}
\usepackage{hyperref}
\usepackage{cite}
\usepackage{enumitem}
\usepackage{bm}
\usepackage[colorinlistoftodos]{todonotes}
\usepackage{mathtools}
\usepackage{algorithm}
\usepackage{algpseudocode}
\usepackage{xspace}
\usepackage{verbatim}
\usepackage{apxproof}

\allowdisplaybreaks[1]

\catcode`"=\active
\newcommand"[1]{\ensuremath{\sp{(#1)}}}

\newcommand{\diag}{\operatorname{diag}}

\newcommand{\innp}[1]{\left\langle #1 \right\rangle}

\newcommand{\mA}{\mathbf{A}}
\newcommand{\mB}{\mathbf{B}}

\newcommand{\ones}{\mathds{1}}
\newcommand{\zeros}{\mathbf{0}}
\newcommand{\vx}{\mathbf{x}}

\newcommand{\vy}{\mathbf{y}}
\newcommand{\vyb}{\overline{\mathbf{y}}}
\newcommand{\vz}{\mathbf{z}}

\newcommand{\vu}{\mathbf{u}}

\newcommand{\vbeta}{\bm{\beta}}

\newcommand{\defeq}{\stackrel{\mathrm{\scriptscriptstyle def}}{=}}
\newcommand{\etal}{\textit{et al}.}

\makeatletter
\def\mathcolor#1#{\@mathcolor{#1}}
\def\@mathcolor#1#2#3{%
  \protect\leavevmode
  \begingroup
    \color#1{#2}#3%
  \endgroup
}
\makeatother

\theoremstyle{plain} \numberwithin{equation}{section}
\newtheorem{theorem}{Theorem}[section]

\newtheorem{lemma}[theorem]{Lemma}
\newtheoremrep{lemma}[theorem]{Lemma}
\newtheorem{proposition}[theorem]{Proposition}
\newtheoremrep{proposition}[theorem]{Proposition}

\newtheorem{fact}[theorem]{Fact}
\theoremstyle{definition}
\newtheorem{definition}[theorem]{Definition}

\newcommand{\opt}{\mathrm{OPT}}
\title{Solving Packing and Covering LPs in $\tilde{O}(\frac{1}{\epsilon^2})$ Distributed Iterations\\ with a Single Algorithm and Simpler Analysis}
\author{Jelena Diakonikolas and Lorenzo Orecchia\\
Computer Science Department, Boston University\\
\texttt{\{jelenad, orecchia\}@bu.edu}}
\date{}

\begin{document}
\maketitle

\begin{abstract}
Packing and covering linear programs belong to the narrow class of linear programs that are efficiently solvable in parallel and distributed models of computation, yet are a powerful modeling tool for a wide range of fundamental problems in theoretical computer science, operations research, and many other areas. Following recent progress in obtaining faster distributed and parallel algorithms for packing and covering linear programs, we present a simple algorithm whose iteration count matches the best known $\widetilde{O}(\frac{1}{\epsilon^2})$ for this class of problems. The algorithm is similar to the algorithm of Allen-Zhu and Orecchia \cite{AO-lp-parallel}, it can be interpreted as Nesterov's dual averaging, and it constructs approximate solutions to both primal (packing) and dual (covering) problems. However, the analysis relies on the construction of an approximate optimality gap and a primal-dual view, leading to a more intuitive interpretation. Moreover, our analysis suggests that all existing algorithms for solving packing and covering linear programs in parallel/distributed models of computation are, in fact, unaccelerated, and raises the question of designing accelerated algorithms for this class of problems.  
\end{abstract}

\section{Introduction}\label{sec:intro}

We consider packing and covering linear programs (LPs), i.e., the problems of the form:

\begin{minipage}[t]{.48\textwidth}
\begin{equation}\label{eq:packing}\tag{P}
\max \left\{ \innp{\mathbf{c}, \vx}: \mA\vx \leq \mathbf{b},\; \vx \geq \zeros\right\},
\end{equation}
\vspace{0pt}
\end{minipage}
\begin{minipage}[t]{.48\textwidth}
\begin{equation}\label{eq:covering}\tag{C}
\min \left\{\innp{\mathbf{b}, \vy}: \mA^T \vy \geq \mathbf{c},\; \vy \geq \zeros \right\},
\end{equation}
\vspace{0pt}
\end{minipage}
where $\mA \in \mathbb{R}^{m \times n}, \mA \geq \zeros,$ $\mathbf{b} \in \mathbb{R}^m, \mathbf{b} > \zeros$, and $\mathbf{c} \in \mathbb{R}^n, \mathbf{c}> \zeros$, $()^T$ denotes a matrix transpose, $\zeros$ is an all-zeros vector, and all inequalities are element-wise. Without loss of generality \cite{AwerbuchKhandekar2009,d-luby1993parallel,AO-lp-parallel}, the problems can be considered in their scaled form, so that $\mathbf{c} = \ones$, $\mathbf{b} = \ones$, and $\min_{ij: A_{ij}\neq 0} A_{ij} = 1$, where $\ones$ is an all-ones vector of the appropriate dimension. Moreover, since we are interested in solving (\ref{eq:packing}), (\ref{eq:covering}) approximately, the scaled matrix $\mA$ can be truncated so that its maximum element $\|\mA\|_{\infty}$ is at most $\mathrm{poly}(m,n, 1/\epsilon)$ \cite{d-luby1993parallel,wang2015unified}, where $m, n$ are the numbers of rows and columns of $\mA$, respectively, and $\epsilon$ is a given approximation parameter. From now on, we assume that the problems are stated in such a scaled and truncated form.

Packing and covering LPs, and even general LPs, 
are solvable sequentially in weakly polynomial time with $\log(\frac{1}{\epsilon})$ dependence on the accuracy $\epsilon$, e.g., via generic Interior Point Method (IPM) solvers. However, even the fastest known IPM solver crucially relies on sequential computation and global information. Moreover, the solver's work depends super-linearly on the input -- the algorithm runs in $\widetilde{O}((N + n^2)\sqrt{n}\log(1/\epsilon))$ time, where $N$ is the number of non-zero elements in the constraint matrix $\mA$ \cite{lee2015efficient}. On the other hand, as packing and covering LPs are equivalent to zero-sum matrix games, their $(1+\epsilon)$-approximate solutions can be obtained in $\widetilde{O}(\frac{\|\mA\|_{\infty}}{\epsilon})$ parallel iterations, each with $O(N)$ total work and $O(\log(N))$ depth, using e.g., the techniques of Nesterov~\cite{nesterov2005smooth} and Nemirovski~\cite{Mirror-Prox-Nemirovski}. However, as already discussed, even in the scaled and truncated form $\|\mA\|_{\infty}$ in general can only be assumed to depend polynomially on problem parameters $m, n, \epsilon$, leading to the overall super-linear computation.

Our focus is on algorithms with poly-logarithmic (in $N, \epsilon$) iteration count and linear in $N$ work per iteration, with at most $\log(N)$ depth. The price paid for this small number of iterations and overall near-linear work is polynomial dependence on the approximation parameter $\epsilon$. Within this category, we make the following distinction between parallel and distributed models of computation. In a distributed model, communication is represented by a bipartite graph in which there is a vertex associated with each variable $j$ and each constraint $i$, and an edge between them if and only if $j$ appears in $i$ with a non-zero coefficient \cite{kuhn2006price}. Information can be exchanged only over the edges of the graph. In contrast, in a parallel model, the memory is shared and complete (global) information can be accessed (as long as computation is performed in e.g., log-depth). For example, Young's algorithm for mixed packing and covering \cite{dc-young2001sequential} is parallel but not distributed, as it requires computing the sum of exponentials of constraint slacks over \emph{all the constraints}, in each iteration.  

Algorithms that fall into the described category have been known since the early 90s, starting with the parallel LP solver of Luby and Nisan \cite{d-luby1993parallel} that runs in ${O}(\frac{\log(n)\log(m/\epsilon)}{\epsilon^4})$ iterations. While the result of Luby and Nisan \cite{d-luby1993parallel} was extended to various settings, including distributed computation model \cite{kuhn2006price,AwerbuchKhandekar2009} and more general mixed packing and covering in both parallel \cite{dc-young2001sequential} and distributed \cite{manshadi2013distributed} settings, until recently there were no improvements on the $\widetilde{O}(\frac{1}{\epsilon^4})$ iteration count from \cite{d-luby1993parallel}. 
Recently, this bound was overcome in the work of Allen-Zhu and Orecchia \cite{AO-lp-parallel} and Mahoney \etal \cite{mahoney2016approximating}, in distributed and parallel settings, respectively, both leading to the $\widetilde{O}(\frac{1}{\epsilon^2})$ iteration count.\footnote{Allen-Zhu and Orecchia \cite{AO-lp-parallel} in fact claimed a $\widetilde{O}(\frac{1}{\epsilon^3})$ bound, however, minor modifications to their algorithm and the corresponding analysis produce the $\widetilde{O}(\frac{1}{\epsilon^2})$ bound.} 

We present a simple algorithm that is similar in spirit to the algorithm of Allen-Zhu and Orecchia \cite{AO-lp-parallel}, but results from a different regularization, which may be of independent interest. The analysis is greatly simplified compared to \cite{AO-lp-parallel}, and has a clear, intuitive interpretation as reducing optimality gap with rate $\frac{1}{k}$, where $k$ is the iteration count. Since the algorithm can be fully analyzed as (unaccelerated) dual averaging, the possibility of designing accelerated algorithms for this class of problems remains open.    
%
%

\subsection{Related Work}
There is a long line of work on packing and covering LPs \cite{kuhn2006price,d-luby1993parallel,d-bartal1997global,dc-young2001sequential,d-papadimitriou1993linear,c-plotkin1995fast,c-fleischer2000approximating,c-garg2007faster,AwerbuchKhandekar2009,AO-lp-parallel,manshadi2013distributed,mahoney2016approximating,arora2012multiplicative,allen2015nearly,koufogiannakis2014nearly,bartal2004fast}. Among them, the first distinction can be made between width-dependent and width-independent algorithms. Width-dependent algorithms have iteration count with super-poly-logarithmic (typically linear or quadratic) dependence on the matrix width $\|\mA\|_{\infty}$. Such algorithms include, e.g.,  (i) the classical work of Plotkin, Shmoys, and Tardos \cite{c-plotkin1995fast}  and a more recent work of Arora, Hazan, and Kale \cite{arora2012multiplicative} that both require only oracle access to the matrix $\mA$, and (ii) more advanced optimization techniques of Nesterov \cite{nesterov2005smooth}, Nemirovski \cite{Mirror-Prox-Nemirovski}, and Bienstock and Iyengar \cite{bienstock2004faster} that leverage explicit knowledge of the matrix $\mA$.

Closer to our work are the algorithms with poly-logarithmic dependence (or no dependence at all) on the matrix width $\|\mA\|_{\infty}$, also known as width-independent algorithms. The work on width-independent algorithms was initiated by Luby and Nisan in \cite{d-luby1993parallel} and extended by Bartal, Byers, and Raz \cite{d-bartal1997global}, providing a parallel algorithm running in $O(\frac{\log(n)\log(m/\epsilon)}{\epsilon^4})$ iterations. Subsequently, similar iteration count was obtained in a more general distributed setting in \cite{kuhn2006price,AwerbuchKhandekar2009}. Moreover, a substantial progress was made for a more general class of positive LPs -- namely, mixed packing and covering LPs -- starting with both sequential and parallel algorithms of Young \cite{dc-young2001sequential}. In the parallel setting, the fastest algorithm is due to Mahoney \etal \cite{mahoney2016approximating}, which solves either pure packing or pure covering LP as a special case of mixed packing and covering LP in $O(\frac{\log(n)\log(n/\epsilon)\log(m\epsilon)}{\epsilon^2})$ iterations\footnote{As \cite{mahoney2016approximating} solves pure packing and covering LPs as special cases of mixed packing/covering LP feasibility problems, the $\log(n/\epsilon)$ factor in the iteration count is incurred due to a binary search.}, while in the distributed setting the algorithm of Allen-Zhu and Orecchia \cite{AO-lp-parallel} solves both (pure) packing and covering LPs with a single algorithm in $O(\frac{\log^2(N/\epsilon)}{\epsilon^2})$ iterations. 

While this paper focuses on width-independent solvers in the parallel and distributed models, sequential algorithms also exist.
Notably, among sequential algorithms, Koufogiannakis and Young \cite{koufogiannakis2014nearly} provide an algorithm that runs in time $O(N + \frac{\log(n)}{\epsilon^2}(n+m))$, while Allen-Zhu and Orecchia \cite{allen2015nearly} obtain an $O(N\frac{\log(N)\log(1/\epsilon)}{\epsilon})$-time packing LP solver and an $O(N\frac{\log(N)\log(1/\epsilon)}{\epsilon^{1.5}})$-time covering LP solver (subsequently improved to $O(N\frac{\log^2(N/\epsilon)\log(1/\epsilon)}{\epsilon})$ by Wang, Rao, and Mahoney in \cite{wang2015unified}). Obtaining a similar total work in a parallel or distributed computation model is an open question.


\subsection{Notation and Preliminaries}

We assume w.l.o.g. that the accuracy $\epsilon$ is from the interval $[0, 1/4]$.

\paragraph{Notation.} We let $\log$ denote the natural logarithm. Similar to \cite{AO-lp-parallel}, we will use the following notation for the truncated gradient:
\begin{equation}\label{eq:trunc-grad}
T_{\nabla_j f(\vx)} = \begin{cases}
\nabla_j f_{\alpha}(\vx), & \text{ if } \nabla_j f_{\alpha}(\vx) \in [-1, 1],\\
1, & \text{ if } \nabla_j f_{\alpha}(\vx) > 1,
\end{cases}
\end{equation}
where $f_{\alpha}(\cdot)$ is the smoothened packing objective (introduced later in this section).

\paragraph{Convex and Concave Conjugates.}  
The following definitions and facts will be useful in the analysis.

\begin{definition}\label{def:cvx-cncv-conj}
The convex conjugate of a function $\psi:X\rightarrow \mathbb{R}$ is defined as $\psi^*(\vz) = \sup_{\vx\in X}\{\innp{\vz, \vx} - \psi(\vx)\}$. Similarly, concave conjugate of $\psi(\cdot)$ is defined as $\psi^*(\vz) = \inf_{\vx \in X}\{\innp{\vz, \vx} - \psi(\vx)\}$. 
\end{definition}

\begin{fact}
Convex conjugate of a convex function is a convex function. Concave conjugate of a concave function is a concave function.
\end{fact}

For the cases we consider here, $X = \mathbb{R}^n$, and thus we can replace ``$\sup$'' and ``$\inf$'' from  Definition \ref{def:cvx-cncv-conj} by ``$\max$'' and ``$\min$'', respectively. 
The following fact is a simple corollary of Danskin's Theorem \cite{danskin2012theory,bertsekas1971control}:
\begin{fact}\label{fact:danskin}
Let $\psi:X\rightarrow \mathbb{R}$ for a closed, convex set $X$, and let $\psi^*$ be its convex conjugate. If $\psi$ is convex, then $\nabla \psi^*(\vz) = \arg\max_{\vx \in X}\{\innp{\vz, \vx} - \psi(\vx)\}$. Similarly, if $\psi$ is concave and $\psi^*$ is its concave conjugate, then $\nabla\psi^*(\vz) = \arg\min_{\vx \in X}\{\innp{\vz, \vx} - \psi(\vx)\}$.
\end{fact}


\paragraph{Smoothing.} Packing and  covering LPs (\ref{eq:packing}),(\ref{eq:covering}) can be stated jointly as follows:
\begin{equation}\label{eq:saddle-point}\tag{P-C}
\min_{\vx \geq \zeros}\max_{\vy \geq \zeros} \left\{- \innp{\ones, \vx} - \innp{\ones, \vy} + \innp{\mA\vx, \vy}\right\}. 
\end{equation}
While the saddle-point formulation (\ref{eq:saddle-point}) encompasses both the primal (packing) and the dual (covering) problems, in general it is non-smooth, limiting the applicability of first-order methods. To circumvent this issue, we can ``smoothen'' the dual problem by adding a strongly convex function $\phi_1(\vy)$ into the maximization problem (with a negative sign) and focus on solving the resulting primal problem. This is similar to the approach taken in \cite{nesterov2005smooth} for a more general class of functions, and in \cite{AO-lp-parallel} for packing and covering LP.  
Unlike previous work \cite{AO-lp-parallel,AwerbuchKhandekar2009,dc-young2001sequential}, which uses generalized entropy, our choice of the regularizer $\phi_1(\vy)$ will be:
\begin{equation}\label{eq:y-reg}
\phi_1(\vy) = -\innp{\ones, \vy} + \frac{1}{1+\alpha}\sum_{i=1}^m {y_i}^{1+\alpha}.
\end{equation}

Expressing $\phi_1^*(\mA\vx - \ones)$ in closed form, the smoothened saddle-point problem (\ref{eq:saddle-point}) reduces to the following minimization problem over the non-negative orthant:
\begin{equation}\label{eq:smoothened-packing}
\min_{\vx \geq \zeros} f_{\alpha}(\vx) \equiv \min_{\vx \geq \zeros} - \innp{\ones, \vx} + \frac{\alpha}{1+\alpha}\sum_{i=1 }^m \left(\mA\vx\right)^{\frac{1+\alpha}{\alpha}}.
\end{equation}

At an intuitive level, $\phi_1(\vy)$ essentially replaces the covering objective $-\innp{\ones, \vy}$ by $- \frac{1}{1+\alpha}\sum_{i=1}^m {y_i}^{1+\alpha}$ which enjoys better structural properties and closely approximates the original objective for sufficiently small $\alpha$. Namely, for our choice of $\alpha = \frac{\epsilon/4}{\log(nm\|\mA\|_{\infty}/\epsilon)}$ , in the region $\vy \in [0, 1]^n$ (which contains the optimal solution $\vy^*$), $- \frac{1}{1+\alpha}\sum_{i=1}^m {y_i}^{1+\alpha}$ multiplicatively $(1+O(\epsilon))$-approximates the original objective, and is, furthermore, $\alpha$-strongly concave. This implies that, for $\vy$ restricted to the hypercube $\vy \in [0, 1]^n$, the resulting minimization problem is $(1/\alpha)$-smooth. While we will not make such a restriction on $\vy$ in order to maintain a valid problem formulation following from (\ref{eq:saddle-point}), it turns out that we will be able to recover a convergence guarantee for the resulting minimization problem that matches that of unaccelerated methods (e.g., gradient descent or mirror-descent) for a $(1/\alpha)$-smooth function.

The following proposition formalizes this intuition.
In particular, it shows that solving (\ref{eq:smoothened-packing}) to multiplicative $(1+\epsilon)$ accuracy suffices to obtain a multiplicative $(1+O(\epsilon))$ solution to the original (non-smoothened) packing problem.

\begin{propositionrep}\label{prop:smoothened-problem}
Let $\opt$ be the optimal value of the packing problem, $\vx \geq \zeros$ be any non-negative vector, $\alpha \leq \frac{\epsilon/4}{\log(mn\|\mA\|_{\infty}/\epsilon)}$, and let $\vx^*_{\alpha}\geq 0$ be the minimizer of $f_{\alpha}(\cdot)$. Then:
\begin{enumerate}
\item \label{it:opt-bounds}(Bounds on $\opt$.) $\frac{1}{\|\mA\|_{\infty}}\leq\opt\leq n$.
\item \label{it:approx-barr} (Approximate barrier property.) If $\mA \vx \leq (1-\epsilon/2)\ones$, then $\frac{\alpha}{1+\alpha}\sum_{i=1 }^m \left(\mA\vx\right)^{\frac{1+\alpha}{\alpha}} \leq \frac{\epsilon}{2} \opt$. Conversely, if, for some $i$, $(\mA\vx)_i \geq 1+\epsilon/2$, then $\frac{\alpha}{1+\alpha}\sum_{i=1 }^m \left(\mA\vx\right)^{\frac{1+\alpha}{\alpha}} > 2\opt$.
\item \label{it:approx-obj} (Approximation guarantee.) If $f_{\alpha}(\vx) \leq (1-\epsilon)f_{\alpha}(\vx_{\alpha}^*)$, then $\mA\vx \leq (1+\epsilon/2)\ones$ and $\innp{\ones, \vx}\geq (1-3\epsilon/2)\opt$.
\end{enumerate}
\end{propositionrep}
\begin{proof}
The proof of Part \ref{it:opt-bounds} is immediate, by observing that $\vx = \frac{1}{n\|\mA\|_{\infty}}\ones$ is packing feasible, while having $\vx_j > 1$ for any coordinate $j$ violates the constraints, as $\min_{ij: A_{ij}>0}A_{ij} = 1$.

For Part \ref{it:approx-barr}, assume first that $\mA\vx\leq (1-\epsilon/2)\ones$. Then, $\forall i$, 
$$(\mA\vx)_i^{\frac{1+\alpha}{\alpha}}\leq (1-\epsilon/2)^{\frac{1}{\alpha}}\leq e^{-\log(M\|\mA\|^2_{\infty}/\epsilon)} < \frac{1}{m\|\mA\|^2_{\infty}}.$$
Thus, $\frac{\alpha}{1+\alpha}\sum_{i=1 }^m \left(\mA\vx\right)^{\frac{1+\alpha}{\alpha}} < \frac{\alpha}{\|\mA\|_{\infty}}\leq \frac{\epsilon}{4}\opt$. On the other hand, if, for some $i$, $(\mA\vx)_i \geq 1+\epsilon/2$, then:
$$
(\mA\vx)_i^{\frac{1+\alpha}{\alpha}}\geq (1+\epsilon/2)^{\frac{1}{\alpha}}\geq \left(\frac{mn\|\mA\|_{\infty}}{\epsilon}\right)^{7/4},
$$
and, therefore, $\frac{\alpha}{1+\alpha}(\mA\vx)_i \geq 2\opt$.

Let $\vx^*$ be an optimal solution to the packing problem.
The second part implies that given an $\vx$ such that $(\mA\vx)_i \geq 1+\epsilon/2$ for some $i$, it must be $f_{\alpha}(\vx)>0$. On the other hand, for $\vx = (1-\epsilon/2)\vx^*$, we have that $f_{\alpha}(\vx) \leq -(1-\epsilon/4)\opt$. Therefore, for the third part of the proposition, $f_{\alpha}(\vx) \leq (1-\epsilon)f_{\alpha}(\vx^*_{\alpha}) < 0$, and thus $\mA\vx \leq (1+\epsilon)\ones$. Further, $-\innp{\ones, \vx} \leq -(1-\epsilon/4)(1-\epsilon)\opt \leq -(1-3\epsilon/2)\opt$, as claimed.
\end{proof}
The proof can be found in the appendix. In the rest of the note, we focus on minimizing $f_{\alpha}(\cdot)$. For simplicity,	 denote $\vx^* = \arg\min_{\vx \geq \zeros} f_{\alpha}(\vx)$.


\section{Algorithm and Convergence Analysis}
The pseudocode of the algorithm is provided in Algorithm \ref{algo:pc-lp} (\textsc{PackingCoveringLP}). All the algorithm steps and parameters will become clear from the analysis, and are only stated here for completeness. It is clear that the algorithm terminates after $K \leq \lceil\frac{\eta}{\gamma}\rceil = O(\frac{\log^2(mn\|\mA\|_{\infty}/\epsilon)}{\epsilon^2})$ iterations.  

\begin{algorithm}
\caption{\textsc{PackingCoveringLP}($\mA, \epsilon$)}
\label{algo:pc-lp}
\begin{algorithmic}[1]
\State Function: $\psi(\vx) = -\innp{\ones, \vx} + \frac{1}{1-\alpha}\sum_{j=1}^n {x_j}^{1-\alpha}$
\State Initialization: $\vx^{(0)} = \frac{1-\epsilon}{n\|\mA\|_{\infty}}$, $\vz^{(0)} = \nabla \psi(\vx^{(0)})$, $\vyb^{(0)}=0$, $k=0$, $A_0 = 1$
\State Parameters: $\alpha = \frac{\epsilon/4}{\log(mn\|\mA\|_{\infty}/\epsilon)}$, $\eta = \frac{1}{\epsilon}$, $\gamma = \frac{\alpha^2\eta}{4}$
\While{$A_k \leq \eta$}
\State $k = k+1$
\State $\vx^{(k)} = \nabla \psi^*(\vz^{(k-1)})$ (i.e., $x_j^{(k)}=(1 + z_j^{(k-1)}/\eta)^{-\frac{1}{\alpha}}, \; \forall j \in \{1,2,...,n\}$)
\State $\vz^{(k)} = \vz^{(k-1)} + \gamma T_{\nabla f(\vx^{(k)})}$
\State $\overline{y}_i^{(k)} = \overline{y}_i^{(k-1)} + (\mA\vx^{(k)})_i^{1/\alpha}$, $\forall i \in \{1, 2,..., m\}$
\State $A_k = A_{k-1} + \gamma$
\EndWhile
\State $\vyb^{(k)} = \vyb^{(k)}/k$
\State\Return $\vx^{(k)}, \vyb^{(k)}$
\end{algorithmic}
\end{algorithm}

At a high level, the analysis follows the general argument of constructing an approximate optimality gap $G_k$ as the difference of an upper bound $U_k$ and a lower bound $L_k$, and showing that $A_k G_k$ is a non-increasing function of iterations $k$ for some increasing sequence $A_k$, as in the general approximate gap framework \cite{thegaptechnique}. However, as $f_{\alpha}(\cdot)$ does not directly fall into any of the standard broad classes of objectives with globally-well-behaved properties (e.g., smooth or Lipschitz continuous)\footnote{In fact, we could make $f_{\alpha}(\cdot)$ be both smooth and Lipschitz continuous by bounding the approximate packing barrier $(\mA\vx)^{\frac{1}{\alpha}}$ by some large enough number. However, this would generally lead to at least linear in $n$ number of iterations.}, we need to resort to a more fine-grained analysis relying on rather local properties of $f_{\alpha}(\cdot)$.

\subsection{Local Smoothness and the Upper Bound}

We start by describing the smoothness properties of $f_{\alpha}(\cdot)$, which will be crucially used in the convergence analysis and will essentially determine the step size. Here, ``smoothness'' is not attained in the classical sense, i.e., we do not have $f_{\alpha}(\vy)\leq f_{\alpha}(\vx) + \innp{\nabla f_{\alpha}(\vx), \vy - \vx} + \frac{L}{2}\|\vy - \vx\|^2$ for some $L \in \mathbb{R}_{++}$. Instead, we will show that $f_{\alpha}(\cdot)$ exhibits a property similar to smoothness in a local sense: under small enough multiplicative updates, second (and higher) order terms in the Taylor approximation of $f_{\alpha}(\cdot)$ are not ``large'' compared to the first-order term. This is formalized in the following lemma. Observe that, due to the different choice of a regularizer, unlike \cite{AO-lp-parallel}, we do not need to require (near-)feasibility of $\vx$ in the packing polytope for this smoothness property of $f_{\alpha}(\cdot)$ to hold.

\begin{lemma}\label{lemma:grad-step}(Local multiplicative smoothness.) 
Let $\vx \geq \zeros$. If $\mB = \diag(\vbeta)$, for vector $\vbeta$ given as: $\beta_j = -c_j\alpha T_{\nabla_j f(\vx)}$, where $c_j \in [0, 1/2)$, $\alpha < 1$, and $T_{\nabla_j f(\vx)}$ is given by (\ref{eq:trunc-grad}), then:
$$
f_{\alpha}(\vx + \mB\vx) - f_{\alpha}(\vx) \leq -\alpha \sum_{j=1}^n c_j(1-2c_j)\nabla_j f_{\alpha}(\vx) T_{\nabla_jf(\vx)} x_j.
$$
\end{lemma}
\begin{proof}
From the Taylor approximation of $f_{\alpha}(\vx + \mB\vx)$:
\begin{align}\label{eq:taylor-approx}
f_{\alpha}(\vx + \mB\vx) \leq f_{\alpha}(\vx) + \innp{\nabla f_{\alpha}(\vx), \mB\vx} + \frac{1}{2}\innp{\nabla^2f_{\alpha}(\vx + t\mB\vx)\mB\vx, \mB\vx},
\end{align}
for some $t\in [0, 1]$. The gradient and the Hessian of $f_{\alpha}(\vx)$ are given by:
\begin{align}\label{eq:grad+hessian}
\nabla_j f_{\alpha}(\vx) = -1 + \sum_{i=1}^m A_{ij}(\mA\vx)_i^{\frac{1}{\alpha}}, \quad \nabla^2_{jk} f_{\alpha}(\vx) =  \frac{1}{\alpha}\sum_{i=1}^m A_{ij}A_{ik} (\mA\vx)_i^{\frac{1}{\alpha}-1}.
\end{align}
Let $\beta_m = \max_j \beta_j$. Then:
$
\nabla^2 f_{\alpha}(\vx + t\mB\vx)\preceq (1+\beta_m)^{\frac{1}{\alpha}-1} \nabla^2 f_{\alpha}(\vx) \preceq 2 \nabla^2 f_{\alpha}(\vx),
$ 
as $(1+\beta_m)^{\frac{1}{\alpha}-1} \leq 2$ is equivalent to $\beta_m \leq 2^{\frac{\alpha}{1-\alpha}}-1 \leq \frac{\alpha}{1-\alpha}$, which is true by the lemma assumptions. Therefore:
\begin{align}
\frac{1}{2}\innp{\nabla^2f_{\alpha}(\vx + t\mB\vx)\mB\vx, \mB\vx} &\leq \frac{1}{\alpha} \sum_{i=1}^m (\mA\vx)_i^{\frac{1}{\alpha}-1} (\mA\mB\vx)_i^{2}\notag\\
\text{(by Cauchy-Schwartz Ineq.) } &\leq \frac{1}{\alpha} \sum_{i=1}^m (\mA\vx)_i^{\frac{1}{\alpha}}\sum_{j=1}^n A_{ij}\beta_j^2 x_j 
= \frac{1}{\alpha} \sum_{j=1}^n \left(\sum_{i=1}^m A_{ij} (\mA\vx)_i^{\frac{1}{\alpha}}\right) \beta_j^2 x_j\notag\\
\text{(by (\ref{eq:grad+hessian})) } &= \frac{1}{\alpha}\innp{\ones + \nabla f_{\alpha}(\vx), \mB^2 \vx}. \label{eq:hessian-term-bound}
\end{align}
As $\beta_j = -c_j\alpha T_{\nabla_jf(\vx)}$, it follows that $c_j\alpha \frac{|\nabla_j f_{\alpha}(\vx)|}{\max\{1, 1+ \nabla_jf_{\alpha}(\vx)\}}\leq |\beta_j| \leq 2c_j\alpha \frac{|\nabla_j f_{\alpha}(\vx)|}{\max\{1, 1+ \nabla_jf_{\alpha}(\vx)\}}$. Combining (\ref{eq:taylor-approx}) and (\ref{eq:hessian-term-bound}):
\begin{align}
f_{\alpha}(\vx + \mB\vx) - f_{\alpha}(\vx) &\leq \innp{\nabla f_{\alpha}(\vx), \mB\vx} + \frac{1}{\alpha}\innp{\ones + \nabla f_{\alpha}(\vx), \mB^2 \vx}\notag\\
&\leq 
-\alpha \sum_{j=1}^n c_j(1-2c_j)\nabla_j f_{\alpha}(\vx) T_{\nabla_jf(\vx)} x_j,\notag
\end{align}
as claimed.
\end{proof}

Since we are focusing on minimizing $f_{\alpha}(\cdot)$, $U_k = f_{\alpha}(\vx^{(k+1)})$ for $\vx^{(k+1)}$ being the solution constructed by the algorithm at the end of iteration $k$ is a valid  upper bound, as long as $\vx^{(k+1)}\geq \zeros$. Lemma \ref{lemma:grad-step} will be used to show that the algorithm steps lead to a sufficiently large decrease in the upper bounds between subsequent iterations.

\subsection{Lower Bound and the Algorithm Steps}

To assess the quality of approximation for a given point $\vx^{(k)}\geq \zeros$, we need a notion of a lower bound to $f_{\alpha}(\vx^*)$. The following lemma constructs one such lower bound.

\begin{lemma}\label{lemma:lower-bound}
Let $\vx^{(0)}, \vx^{(1)},..., \vx^{(k)}$ be a sequence of  points from $\mathbb{R}_+^n$, $a_0, a_1, ..., a_k$ be positive numbers, $A_k = \sum_{s=0}^k a_s$, and let $\phi: \mathbb{R}^n \rightarrow \mathbb{R}$ be a concave function. Then:
$$
f_{\alpha}(\vx^*) \geq L_k \defeq \frac{\sum_{s=0}^k a_s (f_{\alpha}(\vx^{(s)}) - \innp{\nabla f_{\alpha}(\vx^{(s)}), \vx^{(s)}}) + \min_{\vx\geq \zeros}\left\{ \sum_{s=0}^k a_s\innp{T_{\nabla f(\vx^{(s)})}, \vx} - \phi(\vx) \right\} + \phi(\vx^*)}{A_k}.
$$
In particular, if $\phi(\vx) = \psi(\vx) - \innp{\nabla \psi(\vx^{(0)}) - a_0 T_{\nabla f(\vx^{(0)})}, \vx}$ for some continuously-differentiable concave function $\psi:\mathbb{R}^n \rightarrow \mathbb{R}$ and we define $\vz^{(0)} = \nabla \psi(\vx^{(0)})$, $\vz^{(k)} = \vz^{(k-1)} + a_k T_{\nabla f(\vx^{(k)})}$ for $k \geq 1$, then:
$$
L_k = \frac{\sum_{s=0}^k a_s (f_{\alpha}(\vx^{(s)}) - \innp{\nabla f_{\alpha}(\vx^{(s)}), \vx^{(s)}}) +\psi^*(\vz^{(k)}) + \phi(\vx^*)}{A_k}
$$
and
$$
\nabla \psi^*(\vz^{(k)}) = \arg\min_{\vx \geq \zeros}\left\{ \sum_{s=0}^k a_s \innp{T_{\nabla f(\vx^{(s)})}, \vx}-\phi(\vx)\right\}.
$$
\end{lemma}
\begin{proof}
The construction of the claimed lower bound is similar to the general approach from \cite{thegaptechnique}, where in addition we use gradient truncation to account for non-standard smoothness properties of $f_{\alpha}(\cdot)$ (see Lemma \ref{lemma:grad-step}). In particular, by convexity of $f_{\alpha}(\cdot)$, $\forall \vu \geq \zeros, \forall s$: $f_{\alpha}(\vu) \geq f_{\alpha}(\vx^{(s)}) + \innp{\nabla f_{\alpha}(\vx^{(s)}), \vu - \vx^{(s)}}$, and, therefore:
\begin{align}
f_{\alpha}(\vu) \geq \frac{\sum_{s=0}^k a_s (f_{\alpha}(\vx^{(s)}) + \innp{\nabla f_{\alpha}(\vx^{(s)}), \vu - \vx^{(s)})}}{A_k}.\label{eq:trivial-lb}
\end{align}
Recall that $T_{\nabla_j f(\vx)}\leq f_{\alpha}(\vx)$. As $\vu \geq \zeros$, it follows that $\innp{\nabla f_{\alpha}(\vx^{(s)}), \vu} \geq \innp{T_{\nabla f(\vx^{(s)})}, \vu}$. Therefore, subtracting $\frac{1}{A_k}\phi(\vu)$ from both sides of (\ref{eq:trivial-lb}) and taking a minimum over $\vu \geq \zeros$ on the right-hand side of it, we have:
\begin{align}
f_{\alpha}(\vu) - \frac{1}{A_k}\phi(\vu) \geq \frac{\sum_{s=0}^k a_s (f_{\alpha}(\vx^{(s)})-\innp{\nabla f_{\alpha}(\vx^{(s)}), \vx^{(s)}}) + \min_{\vx \geq \zeros}\left\{ \sum_{s=0}^k a_s \innp{T_{\nabla f(\vx^{(s)})}, \vx} - \phi(\vx) \right\}}{A_k}. \notag
\end{align}
Taking $\vu = \vx^*$ in the last inequality yields the claimed lower bound on $f_{\alpha}(\vx^*)$.

For the second part of the lemma, we only need to show that
\begin{equation}\label{eq:psi*-from-lb}
\psi^*(\vz^{(k)}) = \min_{\vx\geq \zeros}\left\{ \sum_{s=0}^k a_s\innp{T_{\nabla f(\vx^{(s)})}, \vx} - \phi(\vx) \right\},
\end{equation}
while the rest of the proof follows from the definition of $L_k$ and by Fact \ref{fact:danskin}. Plugging $\phi(\vx)$ into (\ref{eq:psi*-from-lb}):
\begin{align*}
\min_{\vx\geq \zeros}\left\{ \sum_{s=0}^k a_s\innp{T_{\nabla f(\vx^{(s)})}, \vx} - \phi(\vx) \right\} &= \min_{\vx\geq \zeros}\left\{ \sum_{s=0}^k a_s\innp{T_{\nabla f(\vx^{(s)})}, \vx} - \psi(\vx) + \innp{\nabla \psi(\vx^{(0)})-T_{\nabla f(\vx^{(0)})}} \right\}\\
&= \min\left\{\innp{\vz^{(k)}, \vx} - \psi(\vx)\right\}, 
\end{align*}
which is, by definition, equal to $\psi^*(\vz^{(k)})$.
\end{proof}
We note that instead of $\phi(\vx) = \psi(\vx) - \innp{\nabla \psi(\vx^{(0)}) - a_0 T_{\nabla f(\vx^{(0)})}, \vx}$ we could have used $\phi(\vx) = D_{\psi}(\vx, \vx^{(0)}) + a_0 \innp{T_{\nabla f(\vx^{(0)})}, \vx}$, which is closer to the standard choice of $D_{\psi}(\vx, \vx^{(0)}) = \psi(\vx)-\psi(\vx^{(0)})-\innp{\nabla \psi(\vx^{(0)}), \vx - \vx^{(0)}}$ as a regularizer typically used in first-order methods. The only reason for omitting $-\psi(\vx^{(0)}) + \innp{\nabla \psi(\vx^{(0)}), \vx^{(0)}}$ from $\phi(\vx)$ is to directly work with $\psi^*(\vz^{(k)})$ in the lower bound and avoid bounding unnecessary terms in the gap. The addition of $a_0 \innp{T_{\nabla f(\vx^{(0)})}, \vx}$ is crucial in making the initial gap sufficiently small and it ensures $\nabla \psi^*(\vz^{(0)}) = \vx^{(0)}$. This is specific to the analysis presented here, i.e., such a term does not normally appear in the analysis of standard first order methods (see, e.g., \cite{thegaptechnique}).

Observe that the minimum in the lower bound $L_k$ generates a new point $\vx \geq \zeros$ as its argument. Similar to the standard dual averaging, we will define the sequence of points generated by the algorithm to be the arguments of those minima, that is, $\vx^{(k+1)} = \nabla \psi^*(\vz^{(k)})$. To do so, however, we will need an appropriate choice of $\psi(\cdot)$ that keeps the initial optimality gap sufficiently small and has properties that match well the smoothness of $f_{\alpha}$. The next proposition characterizes $\psi(\vx)$ and, consequently, the sequence of points generated by the algorithm.

\begin{proposition}\label{prop:regularizer-and-its-arg}
Let $\vz^{(0)} = \nabla \psi(\vx^{(0)})$, $\vz^{(k)} = \vz^{(k-1)} + a_k T_{\nabla f(\vx^{(k)})}$ for $k \geq 1$.  If $$\psi(\vx) = \eta\bigg(-\innp{\ones, \vx} + \sum_{j=1}^n \frac{{x_j}^{1-\alpha}}{1-\alpha}\bigg),$$ then $\psi^*(\vz^{(k)}) = - \frac{\eta \alpha}{1-\alpha}\sum_{j=1}^n (1 + z_j^{(k)}/\eta)^{-\frac{1-\alpha}{\alpha}}$, and $\nabla_j \psi^*(\vz^{(k)}) = (1 + z_j^{(k)}/\eta)^{-\frac{1}{\alpha}}$.
\end{proposition}
\begin{proof}
Follows directly by setting the derivative of $ \innp{\vz^{(k)}, \vx}-\phi(\vx)$ with respect to $\vx$ equal to zero, and solving for $\vx$.
\end{proof}

\subsection{The Gap Decrease and Convergence}

By constructing the upper bound and the lower bound, we have fully specified the algorithm and the approximate optimality gap, modulo specifying some of the parameters. In particular, from previous two subsections, using Proposition \ref{prop:regularizer-and-its-arg}, the approximate optimality gap $G_k = U_k - L_k$ is given as:
\begin{align}\label{eq:the-gap}
G_k = \frac{f_{\alpha}(\vx^{(k+1)}) - \sum_{s=0}^k a_s (f_{\alpha}(\vx^{(s)})-\innp{\nabla f_{\alpha}(\vx^{(s)}), \vx^{(s)}}) - \psi^*(\vz^{(k)}) - \phi(\vx^*)}{A_k}.
\end{align}

The overview of the rest of the convergence analysis is as follows. The main goal is to show that $A_k G_k -A_{k-1}G_{k-1 }\leq 0$ (Lemma \ref{lemma:main-lemma}) for some sufficiently fast growing $A_k$, so that $G_k \leq \frac{A_0}{A_k}G_0$. Recall that $f_{\alpha}(\vx^{(k+1)})-f_{\alpha}(\vx^*) \leq U_k - L_k = G_k.$ Thus, we also need to show that $G_0$ is sufficiently small compared to $f_{\alpha}(\vx^*)$. In particular, if $G_0 = O(1)f_{\alpha}(\vx^*)$ (Lemma \ref{lemma:initial-gap}), we immediately get that once $A_k \geq \frac{1}{\epsilon}A_0$, we have obtained a $(1-O(\epsilon))$-approximate solution, i.e., $f_{\alpha}(\vx^{(k)}) \leq (1-O(\epsilon))f_{\alpha}(\vx^*)$. We start by bounding the initial gap in Lemma \ref{lemma:initial-gap}, and then the rest of the convergence analysis will consist of proving the Main Lemma (Lemma \ref{lemma:main-lemma}).

\begin{lemma}\label{lemma:initial-gap} (Initial gap.) 
Let $\eta = \frac{1}{\epsilon}$, $a_0 = 1$, $\vx^{(0)} = \frac{1-\epsilon}{n\|\mA\|_{\infty}}\ones$. Let $\opt = \innp{\ones, \vx^*}$, where $\vx^*$ is the minimizer of $f_{\alpha}(\cdot)$. Then: 
$
G_0 \leq 2 \opt.
$
\end{lemma}
\begin{proof}
From the second part of Proposition \ref{prop:regularizer-and-its-arg}, $\vx^{(1)} = \nabla \psi^*(\vz^{(0)}) = \vx^{(0)}$. Therefore:
\begin{align}
G_0 &= f_{\alpha}(\vx^{(1)}) - f_{\alpha}(\vx^{(0)}) + \innp{\nabla f_{\alpha}(\vx^{(0)}), \vx^{(0)}} - \psi^*(\vz^{(0)}) - \phi(\vx^*)
\leq - \psi^*(\vz^{(0)}) - \phi(\vx^*),\notag
\end{align}
as for $\vx^{(0)} = \frac{1-\epsilon}{n\|\mA\|_{\infty}}\ones$, $\nabla f_{\alpha}(\vx^{(0)}) < 0$ (see Proposition \ref{prop:smoothened-problem}). Using Proposition \ref{prop:regularizer-and-its-arg}:
\begin{align}\label{eq:-psi*-z-0}
-\psi^*(\vz^{(0)}) = \frac{\eta\alpha}{1-\alpha}\sum_{j=1}^n (x^{(0)}_j)^{1-\alpha}\leq \frac{\eta\alpha(1+\epsilon)}{1-\alpha}\innp{\ones, \vx^{(0)}} \leq \frac{1}{2}\opt.
\end{align}
It remains to bound $-\phi(\vx^*) = - \psi(\vx^*) + \innp{\nabla\psi(\vx^{(0)})-\nabla f_{\alpha}(\vx^{(0)}), \vx^*}$. As $\vx^* \leq (1+\epsilon)\ones$ (by Proposition \ref{prop:smoothened-problem}), it follows that $\frac{(x^*_j)^{1-\alpha}}{1-\alpha}\geq x^*_j$, $\forall j$, and thus $-\psi(\vx^*) \leq 0$. Further, observe that:
\begin{align}
\nabla_j \psi(\vx^{(0)}) = \eta \left(-1 + \left(x^{(0)}_j\right)^{-\alpha}\right) = \frac{1}{\epsilon}\left(-1 + \left(\frac{n\|\mA\|_{\infty}}{1-\epsilon}\right)^{\frac{\epsilon/4}{\log(mn\|\mA\|_{\infty}/\epsilon)}}\right) \leq \frac{1}{2}.\notag
\end{align}
Finally, as $\nabla f_{\alpha}(\vx^{(0)})\geq -\ones$, we have:
\begin{align}
-\phi(\vx^*) \leq \innp{\nabla \psi(\vx^{(0)})-\nabla f_{\alpha}(\vx^{(0)}), \vx^*} \leq \frac{3}{2}\innp{\ones, \vx^*} = \frac{3}{2}\opt,\notag
\end{align}
and, combining with (\ref{eq:-psi*-z-0}), we get the final bound $G_0 \leq 2 \opt$, as claimed.
\end{proof}

\begin{lemma}\label{lemma:main-lemma}(Main Lemma.) Let $k \geq 1$. If $G_{k-1} \leq 2\opt$, then $A_k G_k \leq A_{k-1}G_{k-1}$.
\end{lemma}

Assuming that Lemma \ref{lemma:main-lemma} holds, we can apply it inductively starting with the initial gap bound from Lemma~\ref{lemma:initial-gap} to prove the following convergence result for Algorithm \ref{algo:pc-lp}.

\begin{theorem}\label{thm:packing}
Algorithm \ref{algo:pc-lp} (\textsc{PackingCoveringLP}) produces a solution $\vx^{(K)}$ such that $\mA\vx^{(K)}\leq (1+\epsilon)\ones$ and $\innp{\ones, \vx^{(K)}} \geq (1-3\epsilon)\opt$.
\end{theorem}
\begin{proof}
Observe first that the algorithm ends with $K$ such that $A_{K-1} \leq \eta \leq A_{K}$. Using Lemma~\ref{lemma:initial-gap} and applying Lemma~\ref{lemma:main-lemma} inductively, we get $A_k G_k - A_{k-1}G_{k-1}\leq 0$, $\forall k \leq K$. Therefore, we immediately have $G_K \leq \frac{a_0}{A_K}G_0 \leq \epsilon G_0$. From the definition of $G_K$, applying Lemma~\ref{lemma:initial-gap}, $f_{\alpha}(\vx^{(K+1)})\leq f_{\alpha}(\vx^*) + 2\epsilon \opt \leq (1-2\epsilon)f_{\alpha}(\vx^*)$. The rest of the proof follows by applying Proposition \ref{prop:smoothened-problem}.
\end{proof}
\paragraph{Proof of the Main Lemma.} 
The plan for proving the Main Lemma is as follows. First, we show that under the inductive hypothesis that $G_{k-1}\leq 2\opt$ (which is true initially by Lemma \ref{lemma:initial-gap}), $\vz^{(k-1)}$ cannot get ``too small'' (Proposition \ref{prop:small-enough-z-k}). We then use Proposition \ref{prop:small-enough-z-k} to determine the largest possible step size $a_k$ that preserves local smoothness of the upper bound from Lemma \ref{lemma:grad-step} (Proposition \ref{prop:step-size}), and we use Proposition \ref{prop:step-size} jointly with Lemma \ref{lemma:grad-step} to show that the upper bound must decrease sufficiently. Recall that we would like to choose $a_k$ to be as large as possible, since $A_K = \sum_{k=0}^K a_k$ and, by Lemma \ref{lemma:main-lemma} (Main Lemma) the rate of growth of $A_K$ determines the rate at which we decrease the gap. The remaining part of the proof is to show that any decrease in the lower bound between iterations $k-1$ and $k$ is dominated by the decrease in the upper bound from Lemma \ref{lemma:upper-bound-decrease} (Lemma \ref{lemma:lower-bound-change}).


\begin{proposition}\label{prop:small-enough-z-k}
If $G_k \leq 2\opt$, $\eta \leq A_k$, and $\epsilon \leq 1/4$, then $\vz^{(k)} \geq -(\epsilon\eta/2) \ones$.
\end{proposition}
\begin{proof}
Observe that, by the definition of $f_{\alpha}(\cdot)$, $U_k \geq -\opt$. Therefore, if $G_k \leq 2\opt$, it must be $L_k \geq - 3\opt$. Further, observe that $\forall \vx \geq 0$, it must be $f_{\alpha}(\vx) - \innp{\nabla f_{\alpha}(\vx), \vx} \leq 0$, as $f_{\alpha}(\vx) - \innp{\nabla f_{\alpha}(\vx), \vx} = \left(\frac{\alpha}{1+\alpha}-1\right)\sum_{i=1}^m (\mA\vx)_i^{\frac{1+\alpha}{\alpha}}$ and $\alpha < 1$. Therefore, using the definition of the lower bound (Lemma \ref{lemma:lower-bound}), $-\frac{1}{A_k}\psi^*(\vz^{(k)})\geq -3\opt$. As $\eta \leq A_k$ and (by Proposition \ref{prop:regularizer-and-its-arg}) $\psi^*(\vz^{(k)}) = -\frac{\eta\alpha}{1-\alpha}\sum_{j=1}^n (1 + z^{(k)}_j/\eta)^{-\frac{1-\alpha}{\alpha}}$:
\begin{align}
&\frac{\alpha}{1-\alpha} \sum_{j=1}^n (1 + z^{(k)}_j/\eta)^{-\frac{1-\alpha}{\alpha}} \leq 3\opt \notag\\
\Rightarrow \quad & \forall j,\; \frac{\alpha}{1-\alpha} (1 + z^{(k)}_j/\eta)^{-\frac{1-\alpha}{\alpha}} \leq 3\opt \notag.
\end{align}
Suppose that $z^{(k)}_j/\eta < 1-\epsilon/2$ for some $j$. Then, $(1 + z^{(k)}_j/\eta)^{-\frac{1-\alpha}{\alpha}} > e^{\frac{\epsilon(1-\alpha)}{\alpha}} \geq \left(\frac{mn\|\mA\|_{\infty}}{\epsilon}\right)^{7/4} \geq \frac{3\opt(1-\alpha)}{\alpha}$, which is a contradiction.
\end{proof}

Using Proposition \ref{prop:small-enough-z-k}, we now determine the value of $a_k$ that ensures we can apply local smoothness from Lemma \ref{lemma:grad-step} to the upper bound. Recall that we would like to make $a_k$ as large as possible to obtain faster decrease in $G_k$, which, by Theorem \ref{thm:packing}, translates into the convergence time of Algorithm \ref{algo:pc-lp}.
\begin{proposition}\label{prop:step-size}
Let $a_k = \frac{\alpha^2 \eta}{4}$. If $G_{k-1}\leq 2\opt$ and $A_{k-1}\leq \eta$, then, $x^{(k + 1)}_j = \left(1-c_j\alpha  T_{\nabla_j f(\vx^{(k)})}\right) x^{(k)}_j$, for all $j$, where $c_j$ is a number such that $c_j \in \left[\frac{1}{16}, \frac{1}{4(1-\epsilon)}\right]$.
\end{proposition}
\begin{proof}
Recall that $x^{(k+1)}_j = \nabla_j \psi^*(\vz^{(k)}) = (1+z^{(k)}_j/\eta)^{-1/\alpha}$. As $z^{(k)}_j  = z^{(k-1)}_j + a_k T_{\nabla_j f(\vx^{(k)})}$:
\begin{equation}\label{eq:generic-change-in-x}
x^{(k+1)}_j = x^{(k)}_j \left(1 + \frac{a_k T_{\nabla_j f(\vx^{(k)})}/\eta}{1 + z^{(k-1)}_j/\eta}\right)^{-{1}/{\alpha}}.
\end{equation}
As $G_{k-1}\leq 2\opt,$ from Proposition \ref{prop:small-enough-z-k}, $1 + z^{(k-1)}_j/\eta \geq 1-\epsilon$. Therefore, $\frac{a_k T_{\nabla_j f(\vx^{(k)})}/\eta}{1 + z^{(k-1)}_j/\eta} \in [-\alpha^2, \alpha^2] \subset (-1, 1)$, and applying Bernoulli's and exponential inequalities:
\begin{equation}
\left(1 - \frac{a_k T_{\nabla_j f(\vx^{(k)})}}{\eta\alpha(1 + z^{(k-1)}_j/\eta)}\right) \leq \left(1 + \frac{a_k T_{\nabla_j f(\vx^{(k)})}/\eta}{1 + z^{(k-1)}_j/\eta}\right)^{-{1}/{\alpha}} \leq \exp\left(1 - \frac{a_k T_{\nabla_j f(\vx^{(k)})}}{\eta\alpha(1 + z^{(k-1)}_j/\eta)}\right).
\end{equation}
From the definition of $\vz^{(k)}$, we have that $z^{(k-1)}_j = \sum_{s=1}^{k-1} a_s T_{\nabla_j f(\vx^{(s)})} + \nabla_j \psi(\vx^{(0)})$, $\forall j$. We have already shown (in the proof of Lemma \ref{lemma:initial-gap}) that $\forall j$, $\nabla_j \psi(\vx^{(0)})\leq 1/2$. As $T_{\nabla_j f(\vx^{(s)})} \leq 1$ and $a_0 = 1$, it follows that $z^{(k-1)}_j \leq A_{k-1}- \frac{1}{2}.$ By the proposition's assumptions, $A_{k-1}\leq \eta$, and, therefore, $1 + z^{(k-1)}_j/\eta \leq 2$.

The rest of the proposition follows by case analysis for $T_{\nabla_j f(\vx^{(k)})} \in [-1, 0]$ and $T_{\nabla_j f(\vx^{(k)})} \in (0, 1]$, and is omitted.
\end{proof}

We can now apply Lemma \ref{lemma:grad-step} to obtain the desired decrease in the upper bound, as follows.

\begin{lemma}(Change in the upper bound.) \label{lemma:upper-bound-decrease}
Let $a_k = \frac{\alpha^2 \eta}{4}$. If $G_{k-1}\leq 2\opt$ and $A_{k-1}\leq \eta$, then:
$$
A_k U_k - A_{k-1}U_{k-1} \leq a_k f_{\alpha}(\vx^{(k)}) - A_k \frac{\alpha}{18} \sum_{j=1}^n \nabla_j f_{\alpha}(\vx^{(k)})T_{\nabla_j f(\vx^{(k)})}x^{(k)}_j. 
$$
\end{lemma}
\begin{proof}
Follows directly by applying Proposition \ref{prop:step-size} and Lemma \ref{lemma:grad-step}, as $U_k = f_{\alpha}(\vx^{(k+1)})$.
\end{proof}
Next, we bound the change in the lower bound, which will then suffice to show that $A_kG_k - A_{k-1}G_{k-1}$, completing the proof of the Main Lemma.

\begin{lemma}\label{lemma:lower-bound-change}(Change in the lower bound.) 
Let $a_k = \frac{\alpha^2 \eta}{4}$. If $G_{k-1}\leq 2\opt$ and $A_{k-1}\leq \eta$, then:
$$
A_{k}L_k - A_{k-1}L_{k-1} \geq a_k \left(f_{\alpha}(\vx^{(k)})-\innp{\nabla f_{\alpha}(\vx^{(k)}), \vx^{(k)}}\right) + a_k\innp{T_{\nabla f(\vx^{(k)})}, \vx^{(k)}} - \frac{{a_k}\alpha\eta}{4}\sum_{j}x^{(k)}_j \left(T_{\nabla_j f(\vx^{(k)})}\right)^2.
$$
\end{lemma}
\begin{proof}
By the definition of the lower bound (Lemma \ref{lemma:lower-bound}):
\begin{equation}\label{eq:generic-lb-change}
A_{k}L_k - A_{k-1}L_{k-1} = a_k \left(f_{\alpha}(\vx^{(k)})-\innp{\nabla f_{\alpha}(\vx^{(k)}), \vx^{(k)}}\right) + \psi^*(\vz^{(k)}) - \psi^*(\vz^{(k-1)}).
\end{equation}
The rest of the proof follows by a Taylor approximation of $\psi^*(\vz^{(k)}) - \psi^*(\vz^{(k-1)})$:
\begin{align*}
&\psi^*(\vz^{(k)}) - \psi^*(\vz^{(k-1)}) \geq \innp{\nabla \psi^*(\vz^{(k-1)}), \vz^{(k)}-\vz^{(k-1)}}\\
&+ \frac{1}{2}\min\left\{\innp{\nabla^2\psi^*(\vz^{(k-1)})(\vz^{(k)}-\vz^{(k-1)}), \vz^{(k)}-\vz^{(k-1)}}, \innp{\nabla^2\psi^*(\vz^{(k)})(\vz^{(k)}-\vz^{(k-1)}), \vz^{(k)}-\vz^{(k-1)}} \right\}.
\end{align*}
Recall that $\vz^{(k)} = \vz^{(k-1)} + a_k T_{\nabla f(\vx^{(k)})}$ and $\vx^{(k)} = \nabla\psi^*(\vz^{(k-1)})$. Observe that $\nabla^2\psi^*(\vz^{(k)}) = -\frac{1}{\alpha}\diag\left(\vx^{(k+1)}\right)\cdot \diag\left(\ones + \vz^{(k)}/\eta \right)^{-1}$. Applying Proposition \ref{prop:step-size}, $\nabla \psi^*(\vz^{(k)}) \geq -\frac{2}{\alpha} \diag (\vx^{(k)})$. Therefore:
\begin{align}
\psi^*(\vz^{(k)}) - \psi^*(\vz^{(k-1)}) &\geq a_k\innp{T_{\nabla f(\vx^{(k)})}, \vx^{(k)}} - \frac{{a_k}^2}{\alpha}\sum_{j}x^{(k)}_j \left(T_{\nabla_j f(\vx^{(k)})}\right)^2\notag\\
&= a_k\innp{T_{\nabla f(\vx^{(k)})}, \vx^{(k)}} - \frac{{a_k}\alpha\eta}{4}\sum_{j}x^{(k)}_j \left(T_{\nabla_j f(\vx^{(k)})}\right)^2\notag.
\end{align}
Combining with (\ref{eq:generic-lb-change}), the proof follows.
\end{proof}

We are now ready to complete the proof of the Main Lemma.

\begin{proof}[Proof of Lemma \ref{lemma:main-lemma}]
The proof follows by showing that Lemmas \ref{lemma:upper-bound-decrease} and \ref{lemma:lower-bound-change} imply that $A_k G_k - A_{k-1}G_{k-1}\leq 0$. Applying these two lemmas:
\begin{align}
A_k G_k - A_{k-1}G_{k-1} \leq & a_k \innp{\nabla f_{\alpha}(\vx^{(k)})-T_{\nabla f(\vx^{(k)})},\vx^{(k)}}\notag\\
&+ \frac{{a_k}\alpha\eta}{4}\sum_{j}x^{(k)}_j \left(T_{\nabla_j f(\vx^{(k)})}\right)^2 - A_k \frac{\alpha}{18} \sum_{j=1}^n x^{(k)}_j \cdot \nabla_j f_{\alpha}(\vx^{(k)})T_{\nabla_j f(\vx^{(k)})}
=
\sum_{j=1}^n \sigma_j,\notag
\end{align}
where 
$$\sigma_j = a_k \left(\nabla_j f_{\alpha}(\vx^{(k)})-T_{\nabla_j f(\vx^{(k)})}\right) x_j^{(k)}
+ \frac{{a_k}\alpha\eta}{4}x^{(k)}_j \left(T_{\nabla_j f(\vx^{(k)})}\right)^2 - A_k \frac{\alpha}{18} x^{(k)}_j \cdot \nabla_j f_{\alpha}(\vx^{(k)})T_{\nabla_j f(\vx^{(k)})}.$$
By the choice of parameters, $A_k \geq 1$, and $A_k \frac{\alpha}{18}\geq \frac{\alpha}{18} \geq a_k + \frac{a_k \alpha\eta}{4}$. When $\nabla_j f_{\alpha}(\vx^{(k)}) \leq 1$, then $\nabla_j f_{\alpha}(\vx^{(k)}) = T_{\nabla_j f(\vx^{(k)})}$, and we immediately have $\sigma_j \leq 0$. Suppose $\nabla_j f_{\alpha}(\vx^{(k)}) > 1$. Then $T_{\nabla_j f(\vx^{(k)})} = 1$, and we have:
\begin{align*}
\sigma_j \leq x_{j}^{(k)}\left(a_k \nabla_j f_{\alpha}(\vx^{(k)}) + \frac{a_k\alpha\eta}{4} - \frac{\alpha}{18}\nabla_j f_{\alpha}(\vx^{(k)})\right) \leq 0,
\end{align*}
as $\frac{\alpha}{18} \geq a_k + \frac{a_k \alpha\eta}{4}$.
\end{proof}


%
%
\subsection{Extracting a Covering Solution} 
We now argue that the algorithm also constructs an approximate covering solution. Essentially, our lower bound  can be interpreted as a regularized covering objective. Towards that goal, denote for $k \geq 1$:
\begin{equation}\label{eq:dual-soln-def}
y_i^{(k)} = (\mA\vx^{(k)})_i^{1/\alpha}, \quad \vyb^{(k)} = \frac{\sum_{s=1}^k a_s \vy^{(s)}}{A_k - 1}.
\end{equation}
Using the results established earlier in the section, we have the following result.

\begin{theorem}\label{thm:covering}
Let $\vy^{(k)}, \vyb^{(k)}$ be defined via (\ref{eq:dual-soln-def}) and points $\vx^{(k)}$ constructed by Algorithm \ref{algo:pc-lp} (\textsc{Packing\-CoveringLP}). Then $\mA^T\vyb^{(K+1)}\geq (1-2\epsilon)\ones$ and $\innp{\ones, \vyb^{(K+1)}}\leq (1+4\epsilon)\opt$.
\end{theorem}
\begin{proof}
From Proposition \ref{prop:small-enough-z-k}, $\vz^{(K+1)}\geq -(\epsilon\eta/2)\ones$. Recall that $\vz^{(K+1)} = \sum_{s=1}^{K+1}a_s \nabla f_{\alpha}(\vx^{(s)}) + \nabla \psi(\vx^{(0)})$ and that $\nabla \psi(\vx^{(0)})\leq \frac{1}{2}$ (from the proof of Lemma \ref{lemma:initial-gap}). As $A_{K+1} \geq \eta$, $\nabla f_{\alpha}(\vx^{(s)}) = -\ones + \mA^T\vy^{(s)}$, and $\eta = \frac{1}{\epsilon}$:
\begin{align}
&\frac{1}{A_{K+1}-1}\left(\sum_{s=1}^{K+1}a_s \nabla f_{\alpha}(\vx^{(s)}) + \nabla \psi(\vx^{(0)})\right) \geq - \frac{\epsilon}{2}\cdot \frac{\eta}{A_{K+1}-1}\ones\notag \\
\Rightarrow\quad & \mA^T\vyb^{(K+1)} \geq \left(1 - \frac{\epsilon}{1-\epsilon} - \frac{\epsilon}{2(1-\epsilon)}\right)\ones \geq (1-2\epsilon)\ones. \notag
\end{align}

For the second part of the theorem, as $G_{K+1}\leq 2\epsilon \opt$ and $U_{K+1} \geq -(1+\epsilon/2)\opt$, we have that $L_{K+1} \geq -(1+5\epsilon/2)\opt$. Since $\psi^*(\vz) \leq 0$ and $\phi(\vx^*) \geq 0$, it follows that:
\begin{equation}
\frac{1}{A_{K+1}}\sum_{s=0}^{K+1} a_s\left(f_{\alpha}(\vx^{(s)}) - \innp{\nabla f_{\alpha}(\vx^{(s)}), \vx^{(s)}}\right) \geq -(1+5\epsilon/2)\opt.
\end{equation}
Observe that: 
$$f_{\alpha}(\vx^{(s)}) - \innp{\nabla f_{\alpha}(\vx^{(s)}), \vx^{(s)}} = \left(\frac{\alpha}{1+\alpha}-1\right)\sum_{i=1}^m \left(y_i^{(s)}\right)^{1+\alpha} = -\frac{1}{1+\alpha}\sum_{i=1}^m \left(y_i^{(s)}\right)^{1+\alpha}.$$
Applying Jensen's Inequality:
\begin{equation}
-\frac{1}{1+\alpha} \sum_{i=1}^m \left(\overline{y}_i^{(K+1)}\right)^{1+\alpha} \geq - \frac{1}{1-\epsilon}(1+5\epsilon/2)\opt. 
\end{equation}
It is not hard to verify that for $y \geq \frac{\epsilon^2}{m\|\mA\|_{\infty}}$, $\frac{y^{1+\alpha}}{1+\alpha} \geq y$. Therefore:
\begin{align}
\sum_{i: \overline{y}_i^{(K+1)} \geq \epsilon^2/(m\|\mA\|_{\infty})} \overline{y}_i^{(K+1)} \leq \frac{1}{1+\alpha} \sum_{i=1}^m \left(\overline{y}_i^{(K+1)}\right)^{1+\alpha} \leq \frac{1}{1-\epsilon}(1+5\epsilon/2)\opt. 
\end{align}
As $\sum_{i: \overline{y}_i^{(s)} < \epsilon^2/(m\|\mA\|_{\infty})} \overline{y}_i^{(s)} \leq \frac{\epsilon^2}{\|\mA\|_{\infty}}\leq \epsilon^2 \opt$, we get the final bound:
$$
\innp{\ones, \vyb^{(K+1)}} \leq \left( \frac{1+5\epsilon/2}{1-\epsilon} +\epsilon^2\right)\opt \leq (1+4\epsilon)\opt,
$$
as claimed.
\end{proof}
\bibliographystyle{abbrv}
{\small
\bibliography{references}
}
\end{document}